\documentclass{IEEEtran}

\usepackage{graphicx}
\usepackage{cite,url}

\usepackage{amsmath, amssymb, amsthm, amsfonts, mathrsfs}
\usepackage{algorithmic,algorithm}

\usepackage[subtle,mathspacing=normal]{savetrees}

\newtheorem{theorem}{Theorem}

\newtheorem{lemma}{Lemma}

\newtheorem{remark}{Remark}

\newcommand{\openmat}[1]{\left[\begin{array}{#1}}
\newcommand{\closemat}{\end{array}\right]}

\usepackage{tikz}
\usetikzlibrary{shapes,arrows,positioning,calc}
\tikzset{
block/.style = {draw, fill=white, rectangle, minimum height=2.5em, minimum width=3em},
tmp/.style = {coordinate},
sum/.style= {draw, fill=white, circle, node distance=1cm},
input/.style = {coordinate},
output/.style= {coordinate},
pinstyle/.style = {pin edge={to-,thin,black}}}

\begin{document}
\title{Uncertainty Quantification of Data-Driven Output Predictors in the Output Error Setting} 
\author{Farzan Kaviani, Ivan Markovsky, and Hamid R. Ossareh 
\thanks{Farzan Kaviani and Hamid R. Ossareh are with the Department of Electrical and Biomedical Engineering, University of Vermont, Burlington, VT USA 05405. (e-mail: fkaviani@vum.edu, hossareh@vum.edu)}
\thanks{Ivan Markovsky is with the Catalan Institution for Research and Advanced Studies (ICREA), Pg. Lluis Companys 23, Barcelona, Spain. (e-mail: imarkovsky@cimne.upc.edu)}
}

\maketitle
\begin{abstract}
We revisit the problem of predicting the output of an LTI system directly using offline input-output data (and without the use of a parametric model) in the behavioral setting. Existing works calculate the output predictions by projecting the recent samples of the input and output signals onto the column span of a Hankel matrix consisting of the offline input-output data. However, if the offline data is corrupted by noise, the output prediction is no longer exact. While some prior works propose mitigating noisy data through matrix low-ranking approximation heuristics, such as truncated singular value decomposition, the ensuing prediction accuracy remains unquantified. This paper fills these gaps by introducing two {\it upper bounds} on the prediction error under the condition that the noise is sufficiently small relative to the offline data's magnitude. The first bound pertains to prediction using the raw offline data directly, while the second one applies to the case of low-ranking approximation heuristic. Notably, the bounds do not require the ground truth about the system output, relying solely on noisy measurements with a known noise level and system order. Extensive numerical simulations show that both bounds decrease monotonically (and linearly) as a function of the noise level. Furthermore, our results demonstrate that applying the de-noising heuristic in the output error setup does not generally lead to a better prediction accuracy as compared to using raw data directly, nor a smaller upper bound on the prediction error. However, it allows for a more general {\it upper bound}, as the first upper bound requires a specific condition on the partitioning of the Hankel matrix.
\end{abstract}

\begin{IEEEkeywords}
 Uncertainty quantification, Data-driven control, output prediction error bounds, truncated singular value decomposition de-noising
\end{IEEEkeywords}

\section{Introduction}
In recent years, there has been a growing research interest in data-driven methods for simulation and control \cite{MARKOVSKY202142}. 
One particularly notable framework is based on behavioral systems theory (see \cite{MARKOVSKY202142, tutorial}), 
which, thanks to the so-called \textit{fundamental lemma} \cite{MARKOVSKY202142, willems2005note, MARKOVSKY2023110657}, allows for predictions of the output based on offline input-output data. 
In \cite{markovsky2008data}, these predictors are used to simulate the response of LTI systems, and in  \cite{favoreel1999subspace}, they are used to formulate what is known as the subspace predictive controller. In a similar vein, \cite{coulson2019data} used these predictors to formulate an optimal control problem known as the Data Enabled Predictive Control (DeePC). 
Applications of DeePC are reported in \cite{tutorial}. 

Of practical interest is the situation where the offline data is corrupted by noise, which renders the output predictions inaccurate. To handle noise in the offline data, 
DeePC and similar optimal control approaches implement regularizations within the optimization process to enhance its robustness to noise. 
Another approach to handle noise, as employed, for example, in \cite{9468332},  involves de-noising heuristics such as the Truncated Singular Value Decomposition (TSVD) of the underlying data matrix
\cite{MARKOVSKY202142, eckart1936approximation,markovsky2012low}.
 While TSVD has been employed for de-noising in these references,  \cite{recentPaper,9468332,HO-FD:23}  
showed that the effectiveness of TSVD is not universal and may vary depending on the specific setting and implementation.

Given the successes of these data-driven predictors in research and practical settings, there is interest in the community to formally quantify the prediction accuracy of these predictors. In \cite{Stochastic_con}, confidence regions for data-driven prediction algorithms under the assumption of zero-mean Gaussian noise were proposed. Yet, the proposed confidence regions required either  parameters from the underlying state-space model or the noise-free output initial conditions, hindering their practical application in scenarios where these elements are unknown. Thus, to the best of our knowledge, the problem of uncertainty quantification of these output predictors is still largely open. Furthermore, there is no work that formally compares the prediction performance of the output predictors obtained from raw data as compared to those obtained from the de-noising heuristic based on TSVD.

To fill the above gap, this paper first introduces two {\it upper bounds} on the output prediction error in the case of inexact data. The first upper bound is on the output prediction error when the offline data is  used directly to predict the output. 
The second bound applies to the case when the TSVD method is utilized to preprocess the data before calculating the output prediction. 
In either case, we assume the noise is present only on the output data and not on the input data, as is commonly assumed in control systems literature. In other words, our examinations take place in the output error (OE) setting. 
  We model the output noise as a set-bounded signal with a known worst-case magnitude (as opposed to modeling it as a Gaussian noise as is done in \cite{Stochastic_con}). The upper bounds rely solely on the offline noisy data and the known noise level, which allows us to apply these bounds in a practical setting when the noise-free ground-truth data is not available. However, as we show, the derived bounds are only applicable to situations where the noise is ``sufficiently small" as compared to the underlying signal, as characterized by the minimum singular value of the Hankel data matrix comprised of offline noisy input-output data. The bounds can be used in control applications, for example by robustifying data-driven control algorithms such as DeePC against worst-case prediction errors. 
 
 To compare the effectiveness of these bounds and assess whether the TSVD method improves prediction accuracy in the OE setup, comprehensive numerical simulations are conducted. 
 The results show that both upper bounds are relatively small (and thus effective) when the noise level is sufficiently small. Specifically, they are shown to be monotonically (and linearly) decreasing as a function of the noise level. 
  Furthermore, the results indicate that  applying the de-noising heuristic in the output error setup does not generally lead to a better prediction accuracy as compared to using raw data directly, nor a smaller upper bound on the prediction error.
However, it does enable a more general upper bound on the output prediction error, as the first upper bound relies on a specific condition related to the partitioning of the Hankel data matrix.

The paper is organized as follows. The problem statement is provided in Section  \ref{sec:formulation}. The main results are reported in Section \ref{sec:main}. Numerical illustrations are provided in Section \ref{sec:analysis}. Conclusions and future work are reported in Section \ref{sec:conclusion}.

The notation throughout this paper is as follows. The sets $\mathbb{Z}^+$, $\mathbb{R}$, $\mathbb{R}^n$, and $\mathbb{R}^{n\times n}$ denote the set of non-negative integers, real numbers, $n$-dimensional vectors of real numbers, and $n\times n$ matrices with real entries, respectively. Unless otherwise stated, we use the variable $t \in \mathbb{Z}^+$ to denote the discrete time index. For a sequence of matrices $X_1, \ldots, X_n$ with the same number of columns, we denote $[X_1^\top, \ldots, X_n^\top]^\top$ by \textbf{col}$(X_1, \ldots, X_n)$. The vectors $\mathbf{1}$ and $\mathbf{0}$ denote vectors of all ones and all zeros, respectively, where the dimensionality is inferred from the context. For a signal $u(t) \in \mathbb{R}^m, t=0,\ldots,T-1$, we use boldface to denote the signal as a vector of vectors, i.e., $\mathbf{u} = \mathbf{col}(u(0), u(1), \ldots, u(T-1))$.

\section{Problem Formulation}\label{sec:formulation}
Consider an $n$-th order causal LTI system with input vector ${u(t) \in \mathbb{R}^m}$ and output vector $y(t) \in \mathbb{R}^p$.
Suppose a parametric (e.g., state-space or transfer function) model for the system is not available, but we have collected finite-length, noise-free input-output data, stored in vectors $u_d(t)$ and $y_d(t)$, where $t=0,\ldots,L-1$ and subscript $d$ refers to ``data". Using this data, we construct the Hankel matrix of order $T \in \mathbb{Z}^+$, denoted by $\mathcal{H}$, whose columns consist of length-$T$ input and output sub-trajectories:
\begin{equation}\label{eq:H}
\mathcal{H} = \openmat{c} \mathcal{H}_u \\ \mathcal{H}_y \closemat \in \mathbb{R}^{(m+p)T\times M}
\end{equation}
where $M = L-T+1$ is the number of columns, and
$$
\mathcal{H}_u = \openmat{c c c c}
u_d(0) & u_d(1) & \cdots & u_d(L-T) \\
\vdots & \vdots & \ddots & \vdots \\
u_d(T-1) & u_d(T) & \cdots & u_d(L-1) \vspace{0.15cm} 
\closemat,
$$
$$
\mathcal{H}_y = 
\openmat{c c c c}
y_d(0) & y_d(1) & \cdots & y_d(L-T) \\
\vdots & \vdots & \ddots & \vdots \\
y_d(T-1) & y_d(T) & \cdots & y_d(L-1) 
\closemat.
$$

If the Hankel matrix satisfies the generalized persistency of excitation condition \cite{identifiability}, namely
\begin{equation}\label{eq:rank cond}
\text{rank}(\mathcal{H}) = mT+n,
\end{equation}
then any length-$T$ trajectory \textbf{col}$(\mathbf{u},\mathbf{y})\in\mathbb{R}^{(m+p)T}$ will belong to the column space of $\mathcal{H}$. For later use, we introduce the following notation:
\begin{equation}\label{eq:r}
r \triangleq mT+n.
\end{equation}
Said differently, if \eqref{eq:rank cond} is met, there must exist a (non-unique) vector $g$, such that
\begin{equation}\label{eq:uy}
\mathcal{H} g = \left[\begin{array}{c} \mathbf{u} \\ \mathbf{y} \end{array}\right]
\end{equation}
for any $T$-samples long trajectory $\textbf{col}(\mathbf{u},\mathbf{y})$ of the system. This idea has been used in the literature for the purpose of data-driven simulation and prediction \cite{markovsky2008data}. This is done by partitioning the output trajectory, $\mathbf{y}$, into two parts, one of length $T_p \geq 1$ that serves to implicitly fix the initial condition, and of length $T_f \geq 1$ that serves as the predicted output: $\mathbf{y} = \mathrm{\textbf{col}}(\mathbf{y_{ini}}, \mathbf{y_{pred}})$ where $\mathbf{y_{ini}} \in \mathbb{R}^{pT_p}$ and  $\mathbf{y_{pred}} \in \mathbb{R}^{pT_f}$, and $T_p+T_f = T$. Similarly, partition the input as $\mathbf{u} = \mathrm{\textbf{col}} (\mathbf{u_{ini}}, \mathbf{u_{pred}})$, and accordingly the Hankel matrix as $\mathcal{H}_u = \mathrm{\textbf{col}}(U_p, U_f)$ and $\mathcal{H}_y = \mathrm{\textbf{col}}(Y_p, Y_f)$. We can then express \eqref{eq:uy} as:
\begin{equation}\label{eq:partition}
\left[\begin{array}{c}U_p\\U_f\\Y_p\\Y_f\end{array}\right] g = \left[\begin{array}{c}\mathbf{u_{ini}}\\\mathbf{u_{pred}}\\\mathbf{y_{ini}}\\\mathbf{y_{pred}}\end{array}\right]
\end{equation}
As discussed in \cite{MARKOVSKY202142}, if $\mathbf{col}(\mathbf{u_{ini}},\mathbf{y_{ini}})$ is a valid length-$T_p$ trajectory of the system and $T_p \geq \ell$, where $\ell$ is the lag or observability index of the system, one can uniquely solve for the latent initial condition and, thus, uniquely compute $\mathbf{y_{pred}}$. That being said, the vector $g$ is generally not uniquely determined. To solve for $\mathbf{y_{pred}}$, the minimum-norm solution to $g$ is often employed:
\begin{equation}\label{eq:g*}
g^* = \left[\begin{array}{c}U_p\\U_f\\Y_p \end{array}\right]^\dagger\left[\begin{array}{c}\mathbf{u_{ini}}\\\mathbf{u_{pred}}\\\mathbf{y_{ini}} \end{array}\right]
\end{equation}
and so $\mathbf{y_{pred}} = Y_f g^*$ or, equivalently,
\begin{equation}\label{eq:ypred}
\mathbf{y_{pred}} = Y_f \left[\begin{array}{c}U_p\\U_f\\Y_p \end{array}\right]^\dagger\left[\begin{array}{c}\mathbf{u_{ini}}\\\mathbf{u_{pred}}\\\mathbf{y_{ini}} \end{array}\right]
\end{equation}

In many practical situations, while the inputs are exact, the output measurements are corrupted by noise and therefore are not exact\footnote{In some situations, the inputs may be corrupted by noise as well, see for example \cite{eiv-book}. Although the result in this paper can be easily extended to these situations, we will not consider them here for brevity.}. We model the measured output as the clean output plus additive noise: 
\begin{equation}
\label{eq:noise_model_2}
y_m(t) = y(t) + n(t), 
\end{equation}
where $n(t)$ is random and, thus, unknown, but it is magnitude-bounded. That is, there exists a known $N\in\mathbb{R}$ such that:
\begin{equation}
\label{eq:noise_model_1}
\|n(t)\|\leq N, \forall t
\end{equation}
The noise degrades the accuracy of the data-driven predictor in \eqref{eq:ypred} for two reasons: \textit{i}) the noise perturbs every element of the initial output vector, $\mathbf{y_{ini}}$, and \textit{ii}) the noise affects the offline data, $y_d(t)$, and thus perturbs every element of $Y_p$ and $Y_f$. Applying the approach described above to the perturbed matrices gives rise to the following predicted output:
\begin{equation}\label{eq:ypredm}
\mathbf{\tilde{y}_{pred}} = (Y_f+\Delta_2) \begin{bmatrix}
U_p\\
U_f\\
Y_p+\Delta_1\\
\end{bmatrix}^\dagger
\begin{bmatrix}
\mathbf{u_{ini}}\\
\mathbf{u_{pred}}\\
\mathbf{y_{ini}} + \delta\\
\end{bmatrix},
\end{equation}
where $\delta \in \mathbb{R}^{pT_p}$, $\Delta_1 \in \mathbb{R}^{pT_p\times M}$, and ${\Delta_2 \in \mathbb{R}^{pT_f\times M}}$ are unknown vector and matrices whose every element is bounded in magnitude by $N$.

As mentioned in the introduction,  to mitigate the effects of noise, the Hankel matrix is sometimes pre-processed first using a low-rank heuristic to enforce condition \eqref{eq:rank cond}. This is done using the truncated singular value decomposition (TSVD) method \cite[Definition 4]{Vu_2021}. Specifically, the rank-$r$ (with $r$ defined in \eqref{eq:r}) TSVD of $\mathcal{H}$ is given by:
$$
\mathcal{\hat{H}}=\hat{U}\hat{\Sigma}\hat{V}^{\intercal} = \left[\begin{array}{c}\hat{U}_p\\\hat{U}_f\\\hat{Y}_p\\\hat{Y}_f\end{array}\right],
$$
where $\hat{U} \in \mathbb{R}^{(m+p)T\times r}$ and $\hat{V} \in \mathbb{R}^{M\times r}$ are unitary matrices and $\hat{\Sigma}\in \mathbb{R}^{r\times r}$ is a diagonal matrix containing the first $r$ singular values of $\mathcal{H}$. We could now employ \eqref{eq:ypred} to predict the output using sub-blocks derived from $\hat{\mathcal{H}}$:
\begin{equation}\label{eq:ypredtsvd}
\mathbf{\hat{y}_{pred}} = \hat{Y}_f \begin{bmatrix}
\hat{U}_p\\
\hat{U}_f\\
\hat{Y}_p\\
\end{bmatrix}^\dagger
\begin{bmatrix}
\mathbf{u_{ini}}\\
\mathbf{u_{pred}}\\
\mathbf{y_{ini}} + \delta\\
\end{bmatrix}
\end{equation}

We are ready to formally state the main problem addressed in this paper. Consider an $m$-input LTI system of known order $n$ and lag $\ell$, integers $T_p$ and $T_f$ that satisfy $T_p \geq \ell$ and $T_f \geq 1$, and offline data $u_d(t)$ and $y_d(t)$, $t=0,\ldots,L-1$, that satisfy condition \eqref{eq:rank cond} with $T=T_f+T_p$. Suppose we know $u_d(t)$ exactly but not $y_d(t)$. Instead, we have collected a noise-corrupted version of the output as described by Eq. \eqref{eq:noise_model_2}--\eqref{eq:noise_model_1} with a known bound $N$. The problem is to find upper bounds on estimation errors ${\|\mathbf{\tilde{y}_{pred}}-\mathbf{y_{pred}}\|_2}$ and $\|\mathbf{\hat{y}_{pred}}-\mathbf{y_{pred}}\|_2$, where $\mathbf{y_{pred}}, \mathbf{\tilde{y}_{pred}}$, and $\mathbf{\hat{y}_{pred}}$ are defined in \eqref{eq:ypred}, \eqref{eq:ypredm}, and \eqref{eq:ypredtsvd}.

\begin{remark}
    Disturbances and nonlinear effects may also be handled by \eqref{eq:noise_model_2}, \eqref{eq:noise_model_1}. This requires modeling the impact of these disturbances and nonlinearities as a magnitude-bounded noise on the output so that the framework can be applied.
\end{remark}

\begin{remark}
    While the Hankel matrix structure in \eqref{eq:H} is commonly used in the literature, it is not the only matrix structure that can be used for the purpose of data-driven prediction. Other possible structures are the Page matrix or the trajectory matrix, see \cite{MARKOVSKY202142}. The results of this paper are applicable to these structures as well. 
\end{remark}

\section{Main Results}\label{sec:main}

\subsection{Preliminaries}

We begin by introducing the following notation to simplify the presentation:
$$
H_1 = \left[\begin{array}{c}U_p\\U_f\\Y_p \end{array}\right], \tilde{H}_1 = \left[\begin{array}{c}U_p\\U_f\\Y_p + \Delta_1 \end{array}\right], \hat{H}_1 = \left[\begin{array}{c}\hat{U}_p\\\hat{U}_f\\\hat{Y}_p  \end{array}\right],  
$$
$$
\tilde{Y}_f = Y_f +\Delta_2,
h = \begin{bmatrix}
\mathbf{u_{ini}}\\
\mathbf{u_{pred}}\\
\mathbf{y_{ini}}\\
\end{bmatrix},
\tilde{h} = \begin{bmatrix}
\mathbf{u_{ini}}\\
\mathbf{u_{pred}}\\
\mathbf{y_{ini}} + \delta\\
\end{bmatrix}
$$
With this notation, $\mathbf{y_{pred}}$ in \eqref{eq:ypred}, $ \mathbf{\tilde{y}_{pred}}$ in \eqref{eq:ypredm}, and $\mathbf{\hat{y}_{pred}}$ in \eqref{eq:ypredtsvd} can be expressed as:
$$
\mathbf{y_{pred}} = Y_f H_1^\dagger h, \;\;\;\;
\mathbf{\tilde{y}_{pred}} = \tilde{Y}_f \tilde{H}_1^\dagger \tilde{h}, \;\;\;\;
\mathbf{\hat{y}_{pred}} = \hat{Y}_f \hat{H}_1^\dagger \tilde{h}
$$

For a rank-$k$ matrix, $A$, with non-zero singular values $\sigma_i(A)$ for $i = 1, \ldots, k$, we denote:
\begin{align}
&\quad\sigma_{\text{min}}(A) \triangleq \min \{ \sigma_1(A), \ldots, \sigma_k(A) \},\nonumber\\
&\quad\sigma_{\text{max}}(A) \triangleq \max \{ \sigma_1(A), \ldots, \sigma_k(A) \}.\nonumber
\end{align}
Additionally, for $k\geq r$, with $r$ defined in \eqref{eq:r}, we define:
\begin{equation}\label{eq:deltaSN}
\begin{aligned}
\delta_{\text{SN}}(A)&\triangleq\sigma_{r}(A) - \sqrt{pT_pM}N 
\\
\sigma_{\text{sq}}(A) &\triangleq \max\left\{\left(\frac{1}{\delta_{\text{SN}}(A)}\right)^2, \left(\frac{1}{\sigma_{\text{min}}(A)}\right)^2\right\}
\end{aligned}
\end{equation}
which will be used in the following subsections. 

We now present two lemmas, one that pertains to bounds on our perturbation matrices, and another that establishes the rank of $H_1$.
\begin{lemma}
\label{lemm:Lemma_1}
The perturbation matrices satisfy:
$$
\|\delta\|_2 \leq \sqrt{pT_p} N
$$
$$\|\Delta_1\|_F \leq \sqrt{pT_pM}N$$
$$\|\Delta_2\|_F\leq \sqrt{pT_f M}N$$ 
\end{lemma}
\begin{proof}
See the Appendix.
\end{proof}
The above lemma allows us to assess and quantify the effects of perturbation by bounding the noise through known elements. 

\begin{lemma}
\label{lemm:H_1 rank}
    Let rank condition \eqref{eq:rank cond} be satisfied. Then we have that $\text{rank}(H_1) = \text{rank}(\mathcal{H}) = r$.
\end{lemma}
\begin{proof}
    See the appendix.
\end{proof}

In the next subsections, we leverage the above lemmas to provide upper bounds on the prediction errors defined earlier.

\subsection{Upper bound on $\|\mathbf{\tilde{y}_{pred}}-\mathbf{y_{pred}}\|_2$}

The following theorem provides an upper bound on the prediction error $\|\mathbf{\tilde{y}_{pred}}-\mathbf{y_{pred}}\|_2$, i.e., when the output prediction is calculated using the noisy offline data directly without any preprocessing.

\begin{theorem} \label{thm:no tsvd}
Suppose $\delta_{\mathrm{SN}}(\tilde{H}_1) > 0$. Then:
\begin{align}
\label{eq:bound_1}
     & \|\mathbf{\tilde{y}_{pred}} -\mathbf{y_{pred}} \|_2 \leq  \sqrt{2}\sigma_{\mathrm{sq}}(\tilde{H}_1)\sqrt{pT_{p}M}N \nonumber\\
     & \quad \times \left( \|\tilde{Y}_f\|_F+\sqrt{pT_{f}M}N\right) \left( \|\tilde{h}\|_2+\sqrt{pT_{p}}N\right)\nonumber\\
     &\quad + \|\tilde{H}_1^{\dagger}\|_F \sqrt{pT_{p}}N\left( \|\tilde{Y}_f\|_F+\sqrt{pT_{f}M}N\right)  \nonumber\\
     &\quad +\sqrt{pT_{f}M}N\|\tilde{H}_1^{\dagger}\tilde{h}\|_F 
\end{align}
\end{theorem}
\begin{proof}
See the Appendix.
\end{proof}
Theorem \ref{thm:no tsvd} is practical in the sense that the upper bound relies solely on known elements and not on the ground truth about the system model or the underlying noise-free output, which may be unknown, making the bound applicable to any dataset that satisfies the $\delta_{\text{SN}}(\tilde{H}_1) > 0$ condition, which can also be verified using known elements. Of course, the system order, $n$, and noise level, $N$, are assumed to be known. See Remark 3 below for a discussion.

As it turns out, the terms $\|\tilde{H}_1^\dagger\|_F$ and $\sigma_{\text{sq}}(\tilde{H}_1)$ in \eqref{eq:bound_1} may grow unbounded as $N$ tends to 0 (i.e., small noise level), rendering the bound too loose to be useful in practice. We now investigate conditions under which this unbounded growth occurs and how it might be avoided. 
Our analysis relies on the rank of the matrices $H_1$ and $\tilde{H}_1$. 
Let the rank of $\tilde{H}_1$ be denoted by $k$, i.e., $\sigma_{\min}(\tilde{H}_1)=\sigma_k(\tilde{H}_1)$. Recall from Lemma \ref{lemm:H_1 rank} that $\text{rank}(H_1)=r$, which we know almost surely satisfies $r \leq k$. We thus consider two cases: $r<k$ and $r=k$. If $r < k$, we have that $\sigma_k(H_1)=0$, and so we get:
\begin{align} \label{eq:big_sigma}
    & |\sigma_{k}(\tilde{H}_1) - \sigma_{k}(H_1)| \leq \|\Delta_1\|_2 \nonumber \Rightarrow \sigma_{\text{min}}(\tilde{H}_1) \leq \|\Delta_1\|_2 \leq \|\Delta_1\|_F  \nonumber\\
    & \quad \Rightarrow \frac{1}{\sigma_{\text{min}}(\tilde{H}_1)} = \sigma_{\text{max}}(\tilde{H}^{\dagger}_1) \geq \frac{1}{\|\Delta_1\|_F},
\end{align}
where the first line is obtained using Weyl's inequality \cite[Proposition 1]{Vu_2021}.
 Eq. \eqref{eq:big_sigma} reveals an issue, which was also alluded to in \cite{Stewart1977}, namely a small $\|\Delta_1\|_F$ can lead to a large value of $\sigma_{\text{max}}(\tilde{H}^{\dagger}_1)$ and, in turn, a large value for $\|\tilde{H}_1^\dagger\|_F$ (this follows from the fact that $\|\tilde{H}_1^\dagger\|_F^2 = \sum \sigma_i^2$). 
 To see the impact of large $\|\tilde{H}_1^\dagger\|_F^2$ on our upper bound, we  further simplify \eqref{eq:big_sigma} using Lemma \ref{lemm:Lemma_1}, which results in:
 $$
 \|\tilde{H}_1^\dagger\|_F \geq \frac{1}{N \sqrt{pT_p M}}
 $$
 Using this expression, we can lower bound $\sigma_{\text{sq}}(\tilde{H}_1)$ as well:
 $$
 \sigma_{\text{sq}}(\tilde{H}_1) \geq \left(\frac{1}{N \sqrt{pT_p M}}\right)^2
 $$
 Now, if $N$ is sufficiently small, the offline data dominates the perturbation to the extent that we can render $\|\Delta_2\|_F$ and $\|\delta\|_2$ negligible in comparison to $\|\tilde{Y}_f\|_F$ and $\|\tilde{h}\|_2$. Of course, this requires $\|{Y}_f\|_F$ and $\|{h}\|_2$ to be non-zero. In such a situation, the right hand side of \eqref{eq:bound_1} is larger than:
 $$
  \frac{\sqrt{2}}{\sqrt{pT_{p}M}N}\|\tilde{Y}_f\|_F \|\tilde{h}\|_2
 $$
 which diverges as $N$ tends to 0, implying that the right hand side of \eqref{eq:bound_1} diverges.
 Thus, if $r<k$ and the noise is small, our upper bound may be too loose to be useful. We can use TSVD to mitigate this problem, as discussed in the next subsection.

 In the second case, where $k=r$, the bound in \eqref{eq:big_sigma} is no longer applicable and so the unbounded growth of $\|\tilde{H}_1\|^\dagger$ caused by small $\|\Delta_1\|_F$ may be avoided. Even though we do not have a proof of this statement, as we show in Section \ref{sec:analysis}, this is indeed the case for all the randomly-generated systems that we considered.

The condition $\delta_{\text{SN}}(\tilde{H}_1) > 0$ can be viewed as a proxy for signal to noise characteristics of the offline data (ergo the subscript $\mathrm{SN}$ where S stands for signal and N for noise). In particular, this condition will be satisfied if the noise level, $N$, is sufficiently small as compared to the underlying noise-free data. We show in Section \ref{sec:analysis} that the bound (i.e., the right hand side) in Theorem \ref{thm:no tsvd} is small if $\delta_{\text{SN}}(\tilde{H}_1)$ is sufficiently large and the rank condition $k = r$ is satisfied.

We reconsider the situation where $N$ is sufficiently small, this time for the case of $k=r$. As before, the offline data dominates the perturbation, but now   \eqref{eq:bound_1} can be  simplified as follows:
\begin{align}
\label{eq:Linear_eq_1}
& \|\mathbf{\tilde{y}}_{f}-\mathbf{y}_{f}\|_2 \leq N\left(\|\tilde{Y}_f\| \sqrt{2}\sigma_{\text{sq}}(\tilde{H}_1)\sqrt{pT_{p}M} \right. \nonumber\\
& \left. + \|\tilde{h}\|_2 + \|\tilde{Y}_f\|_F\|\tilde{H}_1^{\dagger}\|_F\sqrt{T_{p}}+\sqrt{pT_{f}M} \right. \left. + \|\tilde{H}_1^{\dagger}\tilde{h}\|_F\right)
\end{align}
that is, the bound is linear in $N$. This implies the noise level's strong and direct influence on the presented bound. We will further examine this observation in the numerical section. It is important to mention that this linear behavior requires the rank condition $k = r$ to be satisfied.
    
\begin{remark}
    Theorem \ref{thm:no tsvd} utilizes the $r$-th singular value of $\tilde{H}_1$, which necessitates knowledge of the system order $n$, see Eq.~\eqref{eq:r}. In instances where $n$ is not available, an upper bound on $r$ could be employed as a substitute. The same argument can be made regarding the lag $\ell$, since $\ell \leq n$ for any system.
\end{remark}
\subsection{Upper bound on $\|\mathbf{\hat{y}_{pred}}-\mathbf{y_{pred}}\|_2$}
Next, we will find an upper bound on the prediction error, $\|\mathbf{\hat{y}_{pred}}-\mathbf{y_{pred}}\|_2$, i.e., when the output prediction is calculated using the rank-$r$ TSVD of the Hankel matrix.
\begin{theorem} \label{thm:tsvd}
Suppose $\delta_{\text{SN}}(\hat{H}_1)>0$. Then:
\begin{align}
\label{eq:bound_2}
    & \|\mathbf{\hat{y}_{pred}} - \mathbf{y_{pred}}\|_2 \leq \sqrt{2}\left(\|\tilde{Y}_f\|_F + \sqrt{pT_fM}N\right) \nonumber \\
    & \quad \left(\frac{1}{\delta_{\text{SN}}(\hat{H}_1)}\right)^2 \left(\|\hat{H}_1 - \tilde{H}_1\|_F +\sqrt{pT_pM}N\right) \nonumber\\
    & \quad\left(\|\tilde{h}\|_2+\sqrt{pT_p}N\right) + \|\hat{H}_1^\dagger\|_F(\|\tilde{h}\|_2+\sqrt{pT_p}N) \nonumber\\ 
    & \quad \left(\|\hat{Y}_f-\tilde{Y}_f\|_F+\sqrt{pT_fM}N\right) +\|\hat{Y}_f\hat{H}_1^\dagger\|_F\sqrt{pT_p}N 
\end{align}
\end{theorem}
\begin{proof}
See the Appendix.
\end{proof}
Similar to Theorem \ref{thm:no tsvd}, this theorem also relies solely on known elements, i.e., offline noisy data, system order $n$, and noise level $N$, and is applicable to any data that satisfies the $\delta_{\text{SN}}(\hat{H}_1)>0$ condition. Like the previous theorem $\delta_{\text{SN}}(\hat{H}_1)$ captures the noise characteristics of the dataset and as demonstrated in Section IV, the bound is small when $\delta_{\text{SN}}(\hat{H}_1)$ is sufficiently large. Furthermore, this bound does not exhibit the unbounded growth discussed in Eq. \eqref{eq:big_sigma}, so it is more broadly applicable.

Following the same logic as before, if $N$ is sufficiently small, we can simplify the right-hand side of \eqref{eq:bound_2} as follows. Firstly, the terms $\|\tilde{Y}_f\|_F$ and $\|\tilde{h}\|_2$ are non-zero and dominate $\|\Delta_2\|_F$ and $\|\delta\|_2$. Secondly, 
using \cite[Theorem 3]{Vu_2021}, we have that: 
$$
    \|\hat{H}_1 - \tilde{H}_1\|_F 
     \leq \|\Delta_1\|_F\left(2(1+\sqrt{2})\min\left\{\frac{2\|\Delta_1\|_F}{\sigma_r(H_1)},1\right\}+1\right)
$$
and, from Lemma \ref{lemm:Lemma_1}, we know that $\|\Delta_1\|_F \leq \sqrt{pT_pM}N$. Combining these facts, the right-hand side of \eqref{eq:bound_2} becomes:  
\begin{align}
\label{eq:linear_eq_2}
 & \|\mathbf{\hat{y}_{pred}} - \mathbf{y_{pred}}\|_2 \leq N\Biggl(\sqrt{2}\|\tilde{Y}_f\|_F\left(\frac{1}{\delta_{\text{SN}}(\hat{H}_1)}\right)^2\sqrt{pT_pM}\nonumber \\ 
 & \quad \left(2(1+\sqrt{2})\min\left\{\frac{2\|\Delta_1\|_F}{\sigma_r(H_1)},1\right\}+2\right)\|\tilde{h}\|_2 +\|\hat{H}_1^\dagger\|_F\|\tilde{h}\|_2 \nonumber\\
 & \quad \sqrt{T_fM}\left(2(1+\sqrt{2})\min\left\{\frac{2\|\Delta_2\|_F}{\sigma_r(Y_f)},1\right\}+2\right) \nonumber\\
 & \quad +\|\hat{Y}_f\hat{H}_1^\dagger\|_F\sqrt{T_p}\Bigg) 
\end{align}

The right-hand side of this new inequality is linear in $N$ which shows the clear impact of noise level on the presented upper bound. We will explore this further in the next section.

It is worth mentioning that, the observations made on the tightness of this bound with respect to the values of $\delta_{\text{SN}}(\hat{H}_1)$ and $N$ hold true without imposing the restrictive rank conditions required for Theorem \ref{thm:no tsvd}, namely $r< k$.

\begin{remark}
\label{re:EIV_bound}
    The upper bounds presented in Theorem \ref{thm:no tsvd} and Theorem~\ref{thm:tsvd} can be easily extended to the errors-in-variables (EIV) setting, i.e., when both the inputs and outputs are corrupted by noise. In this setting, if we model the input noise similarly to the output noise \eqref{eq:noise_model_2}, \eqref{eq:noise_model_1}, where the noise on each input is also bounded in magnitude by $N$, we obtain the following modifications. For Theorem \ref{thm:no tsvd}, we can show that if $\sigma_r(\tilde{H}_1)-\sqrt{(pT_p+mT)M}N >0$, then we have:
    \begin{align}
     & \|\mathbf{\tilde{y}_{pred}} -\mathbf{y_{pred}} \|_2 \leq  \sqrt{2}\left( \|\tilde{h}\|_2+\sqrt{pT_{p}+mT}N\right) \nonumber\\
     & \quad \max\left\{\left(\frac{1}{\sigma_r(\tilde{H}_1)-\sqrt{(pT_p+mT)M}N}\right)^2, \left(\frac{1}{\sigma_{min}(\tilde{H}_1)}\right)^2\right\} \nonumber \\
     & \quad \times \left( \|\tilde{Y}_f\|_F+\sqrt{pT_{f}M}N\right) \sqrt{(pT_p+mT)M}N  \nonumber\\
     &\quad + \|\tilde{H}_1^{\dagger}\|_F \sqrt{pT_{p} +mT}N \left( \|\tilde{Y}_f\|_F+\sqrt{pT_{f}M}N\right) \nonumber \\ &\quad +\sqrt{pT_{f}M}N\|\tilde{H}_1^{\dagger}\tilde{h}\|_F \nonumber
\end{align}
For Theorem~\ref{thm:tsvd}, we can show that if $\sigma_r(\hat{H}_1)-\sqrt{(pT_p+mT)M}N >0$, then we have that:
    \begin{align}
    & \|\mathbf{\hat{y}_{pred}} - \mathbf{y_{pred}}\|_2 \leq \sqrt{2}\left(\|\tilde{Y}_f\|_F + \sqrt{pT_fM}N\right) \nonumber \\
    & \left(\frac{1}{\sigma_r(\hat{H}_1)-\sqrt{(pT_p+mT)M}N}\right)^2 \left(\|\tilde{h}\|_2+\sqrt{pT_p+mT}N\right)\nonumber\\ 
    & \left(\|\hat{H}_1 - \tilde{H}_1\|_F +\sqrt{(pT_p+mT)M}N\right) + \|\hat{H}_1^\dagger\|_F\nonumber\\
    & (\|\tilde{h}\|_2+\sqrt{pT_p+mT}N) \left(\|\hat{Y}_f-\tilde{Y}_f\|_F+\sqrt{pT_fM}N\right)  \nonumber\\ 
    & +\|\hat{Y}_f\hat{H}_1^\dagger\|_F\sqrt{pT_p+mT}N \nonumber
    \end{align}
\end{remark}
\begin{remark}
     If $\sigma_{r}(H_1)$ is known, we can replace $\sigma_{\text{sq}}(\tilde{H}_1)$ with $\max\left\{\left(\frac{1}{\sigma_{\text{r}}(H_1)}\right)^2, \left(\frac{1}{\sigma_{\text{min}}(\tilde{H}_1)}\right)^2\right\}$ and $\left(\frac{1}{\delta_{\text{SN}}(\hat{H}_1)}\right)^2$ with $\max\left\{\left(\frac{1}{\sigma_{\text{r}}(H_1)}\right)^2, \left(\frac{1}{\sigma_{\text{min}}(\hat{H}_1)}\right)^2\right\}$ in Theorem \ref{thm:no tsvd} and Theorem \ref{thm:tsvd}, respectively. This results in smaller bounds in \eqref{eq:bound_1} and \eqref{eq:bound_2}, and eliminates the requirement for a positive $\delta_{SN}$. However, knowing $\sigma_{r}(H_1)$ requires knowledge on the noise-free offline data so it may not be practical if such information is unavailable.
 \end{remark}
 
\section{Analysis of the bounds}\label{sec:analysis}

In this section, we provide numerical illustrations of the tightness (i.e., smallness) of the upper bounds established in the previous section and the effects of the TSVD method on prediction accuracy. We explore 
different noise scenarios and pinpoint the conditions that result in a small bound. For this analysis, we introduce the notion of the ``relative gap", which measures the percentage difference between the left-hand side and the right-hand side of \eqref{eq:bound_1} and \eqref{eq:enhanced_bound} normalized by the output magnitude:
$$\text{Relative Gap} = \frac{\text{Right-hand side} - \text{Left-hand side}}{\|\mathbf{y_{pred}}\|_2} \times100$$
 A small relative gap means that the bounds are relatively tight, and thus the bounds can be viewed as reasonable approximations of the true prediction error. On the other hand, a large relative gap means that the bounds are too loose to be useful in practice.

We use Monte Carlo experiments of randomly-generated first- and second-order systems, where  we evaluate the effectiveness of the TSVD method on prediction accuracy, the general performance and applicability of the upper bounds, the effects of noise on their tightness, and the conditions where these effects are small.

Each Monte Carlo study comprises of 1,000 randomly-generated {stable} systems. While the process for generating most of the parameters for these systems is identical (as described below) the experiments differ in the way conditions on $T_p$ and $\delta_{\text{SN}}$ are enforced. 

{\bf Random systems:} The random systems are $n$-th order stable discrete-time systems with $p$ outputs and $m$ inputs and are generated using the \texttt{drss} command in MATLAB R2020a. Parameters $n$, $p$, $m$, and $T_f$ are all drawn from discrete uniform distributions, where $n \in \{1,2\}$, $T_f \in \{1,2,3\}$, and $p,m$ are chosen between $1$ and $n$. 
The horizon $T_p$ will be discussed later. 
For each random system, we simulate the system's response to an input signal with elements uniformly distributed between $-1$ and $1$. Using this approach, we collect 100 timesteps of input and output offline data ($L=100$). 
Since the offline data is randomly generated, it is guaranteed to satisfy condition \eqref{eq:rank cond}. As for the online data, the elements of $\mathbf{u_{ini}}$ and $\mathbf{u_{pred}}$, as well as the latent initial state $x_{ini}$, are generated randomly from the uniform distribution between $-1$ and $1$. We then use $\mathbf{u_{ini}}$ and ${x_{ini}}$ to calculate $\mathbf{y_{ini}}$. As for the output noise, we use 50 logarithmically-spaced points between decades $10^{-8}$ and $10^{-3}$ as our noise levels. For each noise level, we corrupt the output data with 100 different random noise realizations. Each noise element will be a uniformly-distributed random number in the interval $(-N,N)$. In total, we have 1,000 different random systems with 5,000 different noise scenarios for each system.

{\bf TSVD prediction accuracy:} We first study the effects of the TSVD method itself on the prediction accuracy. To this end, we generate 1,000 systems as described above with $T_p$ chosen randomly between 1 and 3 for each system. There are no restrictions on $\delta_{\text{SN}}$ for this study. We record $\|\mathbf{\tilde{y}_{pred}}-\mathbf{y_{pred}}\|_2$ and $\|\mathbf{\hat{y}_{pred}}-\mathbf{y_{pred}}\|_2$ (i.e., the left hand sides in \eqref{eq:bound_1} and \eqref{eq:bound_2}) in each scenario, and plot them against each other. The results are presented in Fig. \ref{fig:lhs1_vs_lhs2}, with a 45-degree line for better visual comparison. A logarithmic scale is used for both axes.
\begin{figure}
    \centering
    \includegraphics[width=0.9\columnwidth]{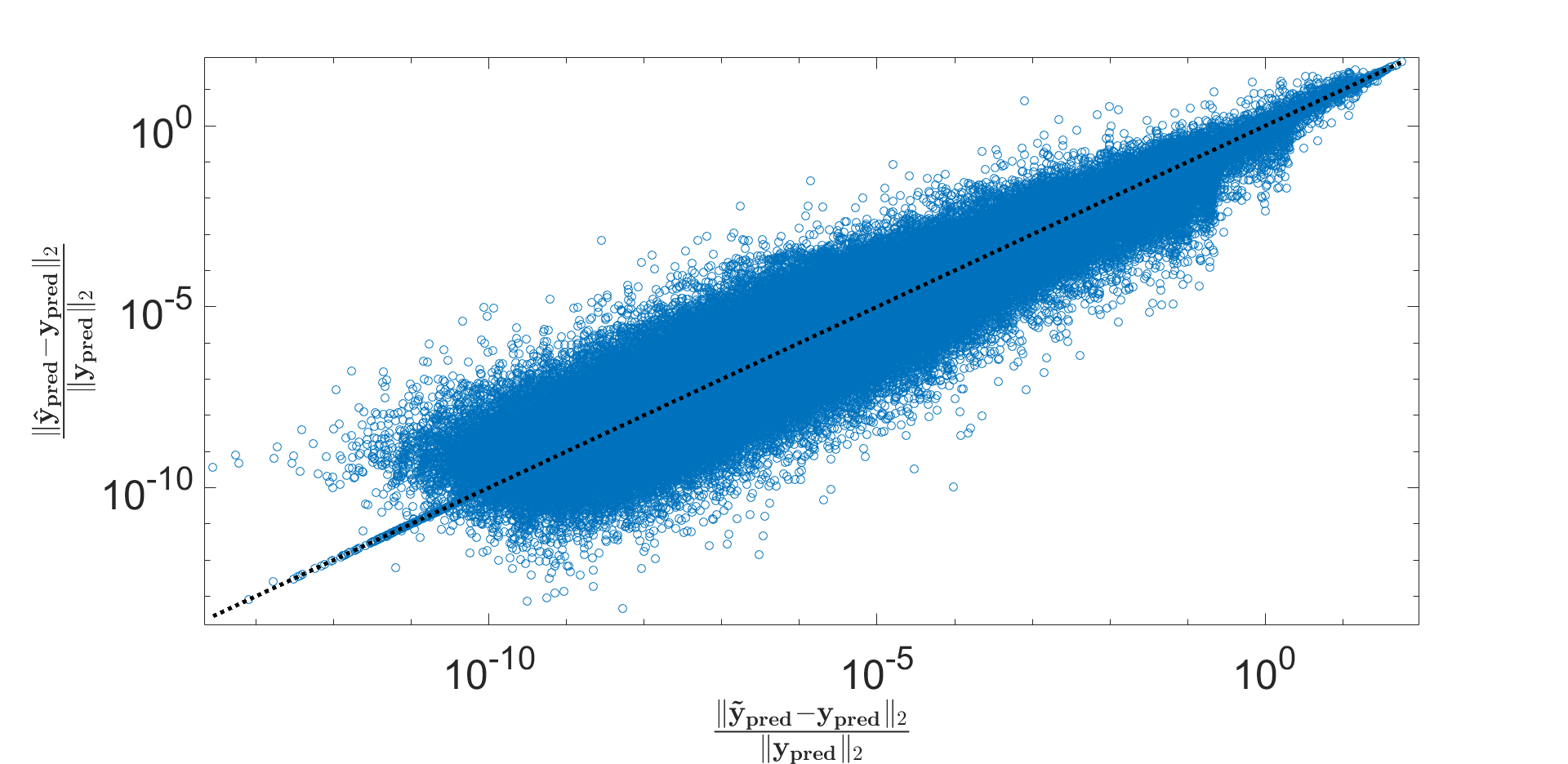}
   \caption{Comparison of { normalized} prediction errors using raw offline data and low-rank approximation of offline data}
   \label{fig:lhs1_vs_lhs2}
\end{figure}
The study reveals that the normalized prediction error remained low for most of the analyzed scenarios, especially when $\delta_{SN}$ had a sufficiently large value. While the TSVD method improved predictions in 52\% of the scenarios, there were instances where it led to poorer predictions. Therefore, we conclude that it cannot be used generically to mitigate noise in OE settings. 
\begin{remark}
    We acknowledge that the application of unstructured rank-r approximation \eqref{eq:ypredtsvd} may not be appropriate for the OE setup \eqref{eq:noise_model_1} because this approach inherently assumes that all elements of the Hankel matrix (i.e., both the inputs and outputs) are perturbed, which is not the case in our scenario. For the OE setup, the Generalized Low Rank Approximation (GLRA) \cite{GLRA} can be considered as a more suitable de-noising heuristic. This is a topic for future research.
\end{remark}
{\bf Comparison of the upper bounds:} Next, we generate 1,000 random systems as described above, but choose $T_p$ such that $\mathrm{rank}(\tilde{H}_1) =r$. Although this condition is only required for the first bound, we enforce it for the second bound as well to allow for a fair comparison between the two bounds. Note that, for the condition $\mathrm{rank}(\tilde{H}_1) =r$ to hold, we must have have that
$$
T_p = \frac{n}{p}.
$$
To show this, notice that $\tilde{H}_1$ has $mT+pT_p$ rows. We know that $\tilde{Y}_p$ is almost surely full row rank due to the measurement noise, i.e., $\mathrm{rank}(\tilde{Y}_p) = pT_p$. Thus, if $\mathcal{H}_u$ is also full row rank, which is the case for our simulations due to the offline input being randomly generated, then $\tilde{H}_1$ is almost surely full row rank as well, which implies that it must have $r$ rows. Therefore, we have that $r = mT+pT_p$, which coupled with condition \eqref{eq:rank cond}, implies that $T_p = \frac{n}{p}$.
As discussed previously, by setting $T_p = \frac{n}{p}$, we may prevent the unbounded growth described in \eqref{eq:big_sigma}. In addition to ensuring that $T_p = \frac{n}{p}$, 
we ensure that $\delta_{\text{SN}} > 0$ for both $\tilde{H}_1$ and $\hat{H}_1$ as required by Theorems \ref{thm:no tsvd} and \ref{thm:tsvd}. 

The recorded relative gaps are plotted against each other in Fig. \ref{fig:relgap1_vs_relgap2}, with a 45-degree line for better visual comparison. Additionally, the {\it average} recorded upper bound values at each noise level for both theorems are plotted against the noise level $N$ in Fig. \ref{fig:rhs1_vs_rhs2} to illustrate the impact of this value on our upper bounds. A logarithmic scale is used for both axes in both figures.
\begin{figure}
    \centering
    \includegraphics[width=\columnwidth]{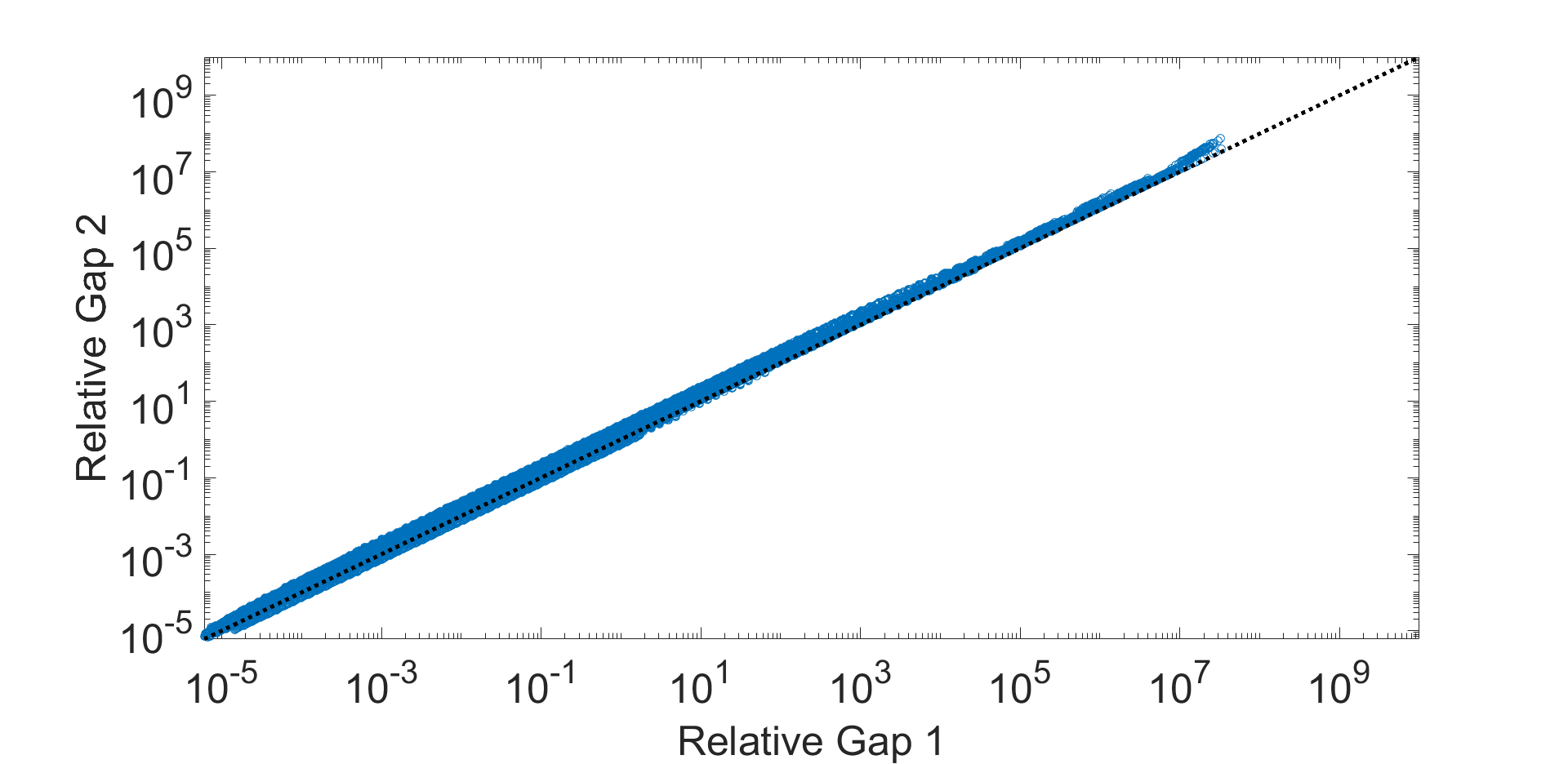}
\caption{Comparison of the relative gaps obtained from Theorem~1 and Theorem~2.}
   \label{fig:relgap1_vs_relgap2}
\end{figure}
\begin{figure}
    \centering
    \includegraphics[width=\columnwidth]{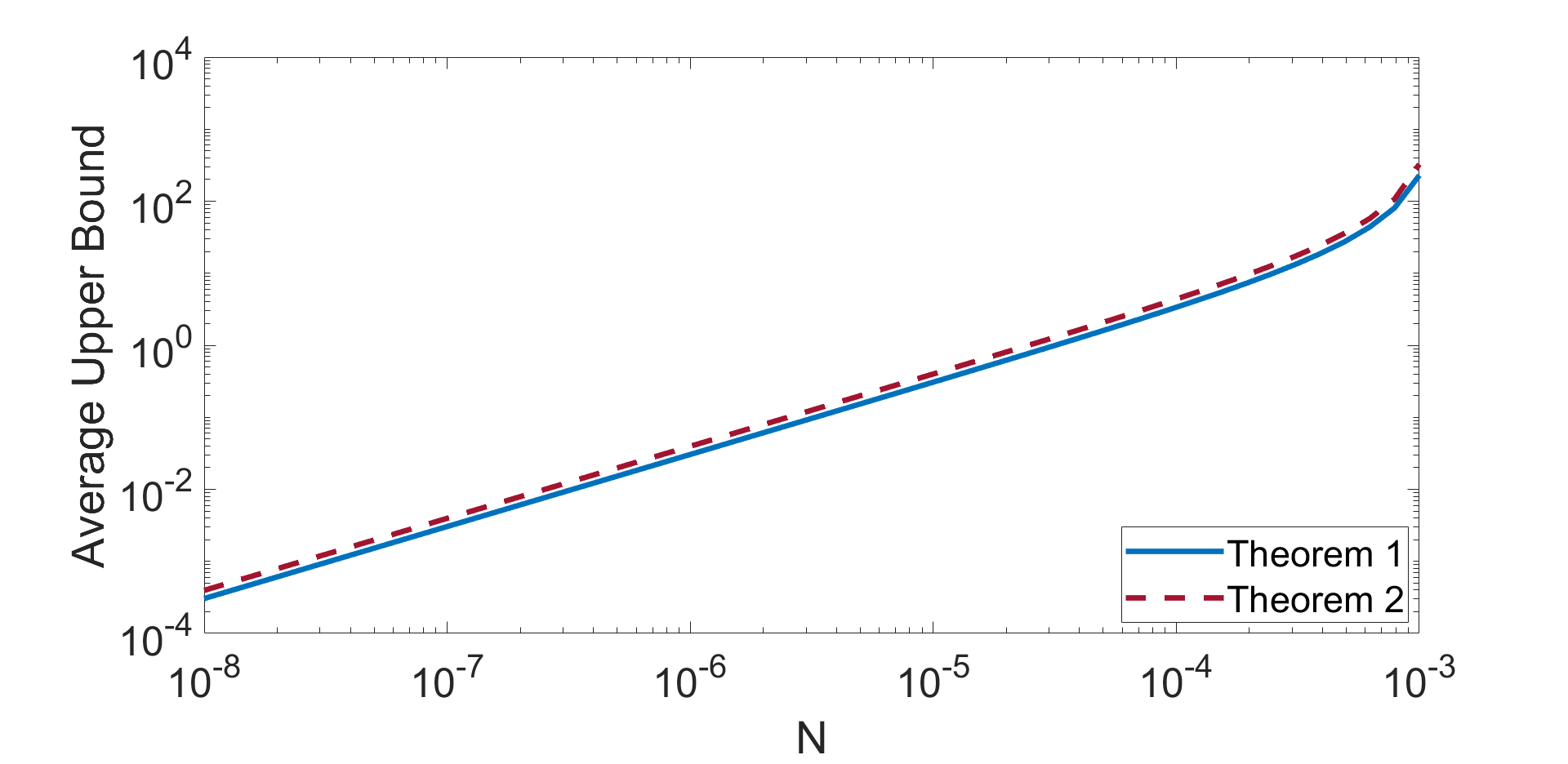}
   \caption{Average upper bound values at each noise level plotted against the noise level $N$.}
   \label{fig:rhs1_vs_rhs2}
\end{figure}
The results illustrated in Fig. \ref{fig:relgap1_vs_relgap2} show that the first upper bound outperformed the second one in most simulations. However, this advantage depended on whether or not $T_p=n/p$, as the relative gap would have had large values for the first bound if this condition was not met (since the considered noise levels are relatively small). Regardless, there were cases where both upper bounds demonstrated large relative gaps, which occurred when $N$ is large. Further, as depicted in Fig. \ref{fig:rhs1_vs_rhs2}, the plots for both upper bounds exhibit a slope of approximately 1, indicating an almost linear relationship between the noise level $N$ and the bounds, which is an observation that is consistent with those made based on inequalities (15) and (17). 

{\bf Enforcing a lower bound on $\delta_{\text{SN}}$:} One of the major reasons for the large gaps mentioned above was the value of $\delta_{\text{SN}}$ being small and close to zero, which resulted in large values for both upper bounds.  Therefore, we repeated the simulations, this time enforcing that $\delta_{\text{SN}}$ was above 0.6 for $\tilde{H}_1$ and $\hat{H}_1$ for the first and second bound, respectively. Additionally, now that we are studying the bounds individually and not against each other, we no longer enforced the $T_p=n/p$ condition for the second bound. In these simulations, we recorded the highest relative gaps (worst-case) in each noise level and plotted a box plot (Fig. \ref{fig:B_box}) to visualize the performance of the bounds for all systems in the worst-case scenarios. We use a logarithmic scale for both axes to provide clearer visualization.
\begin{figure}
   \centering
   \begin{minipage}{\columnwidth}
    \centering
    \includegraphics[width=\columnwidth]{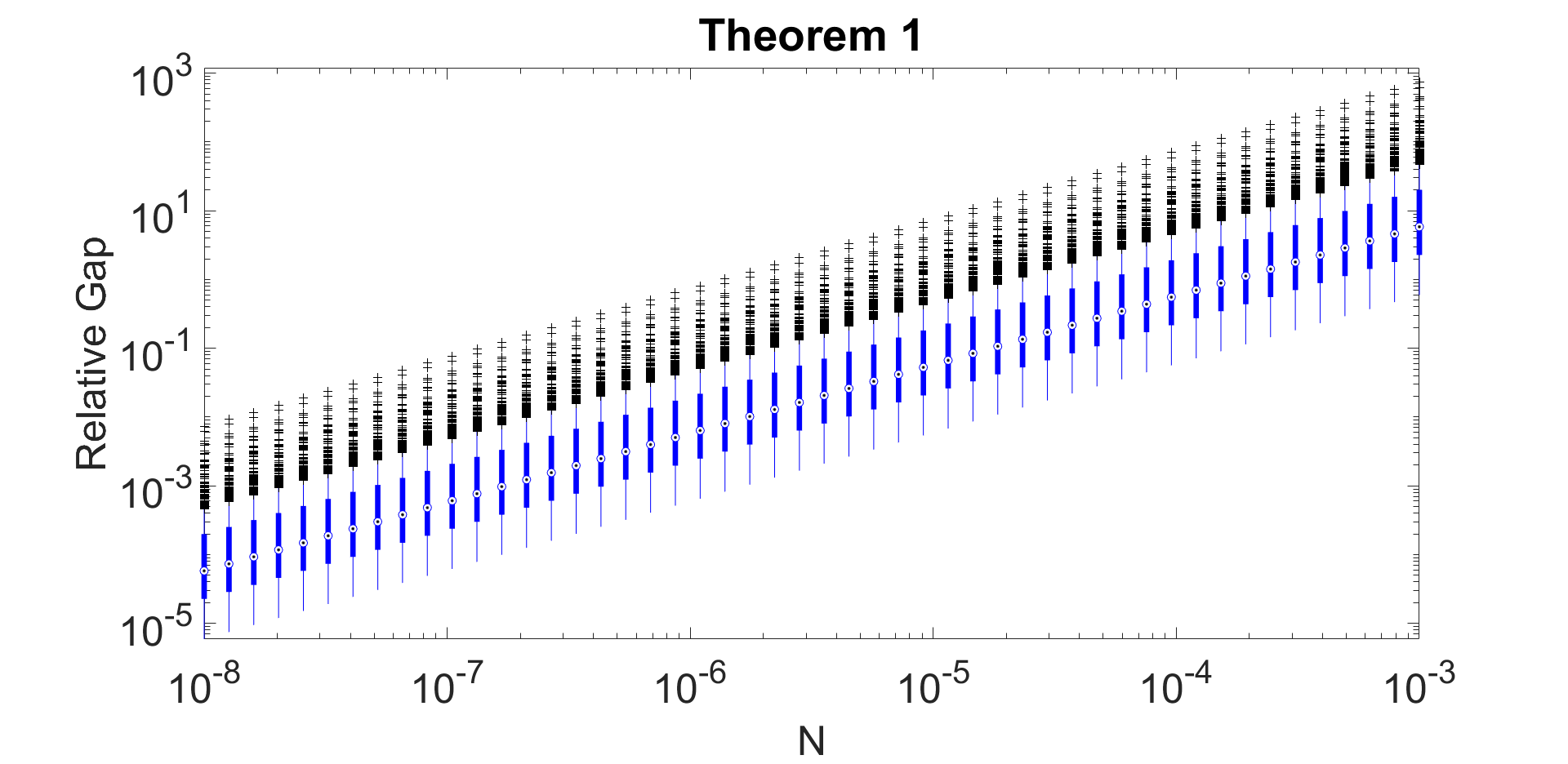}
   \end{minipage}
   \begin{minipage}{\columnwidth}
    \centering
    \includegraphics[width=\columnwidth]{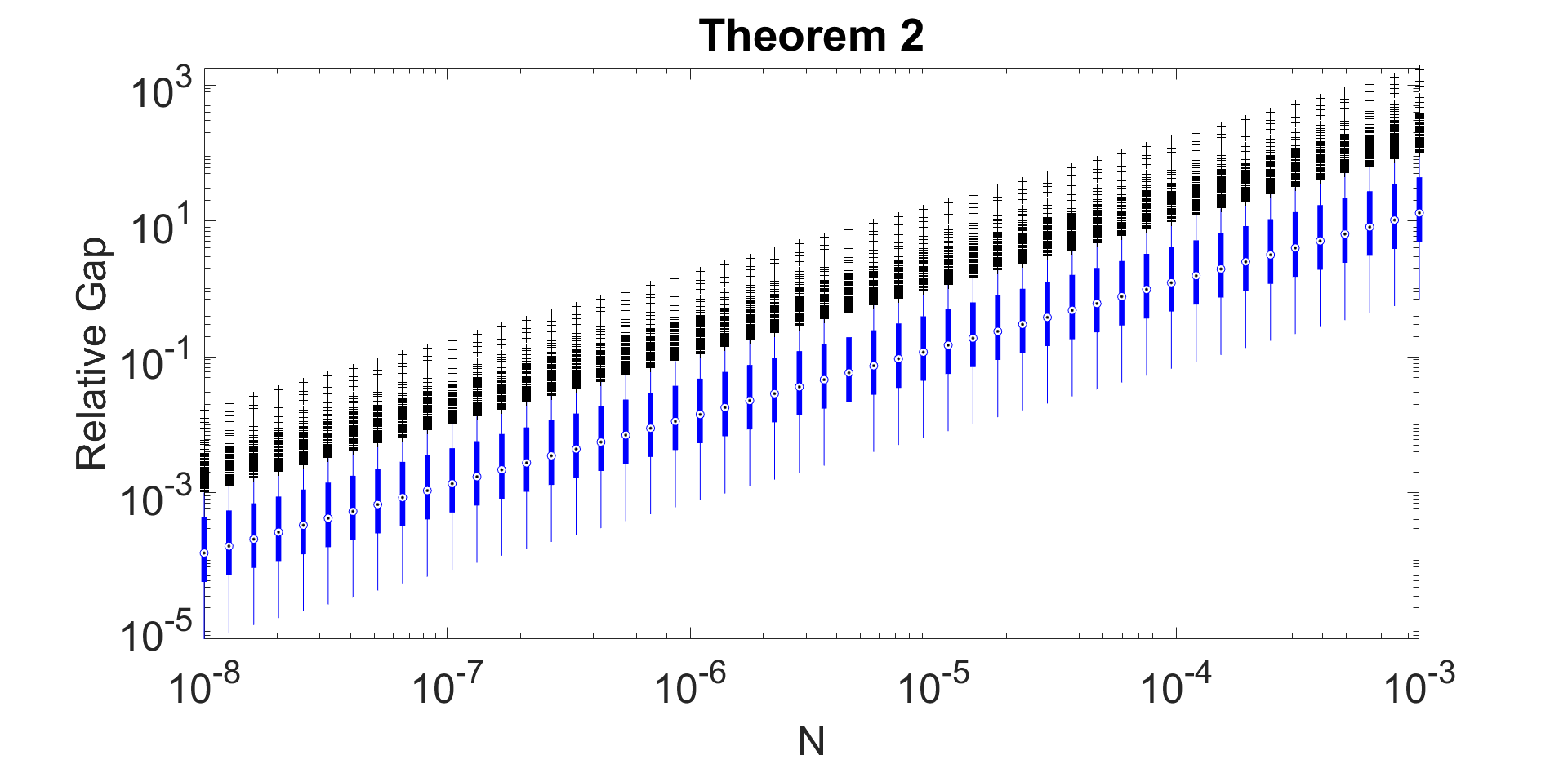}
   \end{minipage}
   \caption{Box plots illustrating the median, 25th, and 75th percentiles (box edges) of relative gap values for both theorems under the influence of the $\delta_{\text{SN}}>0.6$ condition. Whiskers extend to the most extreme non-outlier data points, and outliers are plotted individually with `+' symbols. The top and bottom subplots depict the relative gap values for the first and second upper bounds, respectively.}
   \label{fig:B_box}
\end{figure}
The plots reveal the high accuracy of both upper bounds under the specified $\delta_{\text{SN}}$ condition, with median and average values remaining below 10\% and 25\%, for the first theorem and below 15\% and 50\% for the second one, across all noise levels. It is clear that imposing a minimum $\delta_{\text{SN}}$ threshold across the datasets significantly enhances the tightness of these bounds. Even in instances where large relative gaps were observed, they remained notably smaller compared to those occurring without the $\delta_{\text{SN}}$ condition. Overall, the first theorem performed better in the presence of the $\delta_{\text{SN}}$ condition, while the second theorem had the advantage of not requiring any restrictions on the value of $T_p$.

\section{Conclusions and Future Work}\label{sec:conclusion}

This paper introduced two novel upper bounds on output prediction error, one for direct prediction from raw offline data and another for prediction using the low-rank approximation of the offline data. These upper bounds leverage offline system data and a known noise level denoted as $N$ to bound the error. The results are validated with a series of  Monte Carlo simulations. We demonstrated the effectiveness and precision of our bounds and identified the optimal setting for a close bound. The first bound offers slightly higher accuracy, while the second bound is more versatile and does not impose constraints on the construction of the Hankel matrix. Further, We examined the TSVD method's applicability in mitigating noise effects in output prediction, highlighting its limitations due to inconsistent performance in the OE setup. Future work will study how the theorems can enhance the data-driven methods based on the behavioral approach to systems theory and improve their performance in data-driven control and simulation when dealing with inexact offline data. Finally, we will investigate the application of GLRA in the OE setup.
\section{Appendix}
\subsection{Proof of Lemma \ref{lemm:Lemma_1}}
The matrices $\delta$, $\Delta_1$, and $\Delta_2$ have $pT_p$, $pT_pM$, and $pT_fM$ entries, respectively. By leveraging the definition of the 2-norm and the Frobenius-norm, along with the fact that each element of these matrices is magnitude-bounded by the noise level $N$, we obtain the results.

\subsection{Proof of Lemma \ref{lemm:H_1 rank}}
Suppose the underlying state-space model of the dynamics is given by:
$$
x(t+1) = Ax(t)+Bu(t),
$$
$$
y(t)=Cx(t)+Du(t).
$$
Define the vector of latent states, $X$, and the extended observability and convolution matrices $\mathcal{O}_t$ and $\mathcal{T}_t$ as follows:
$$X=\left[ \begin{array}{cccc}
     x(T_p) & x(T_p+1) & \dots & x(L-T_f)
\end{array} \right]$$
$$\mathcal{O}_{t}=\left[ \begin{array}{cccc}
     C^{\top} & (CA)^{\top} & \dots & (CA^{t-1})^{\top}
\end{array} \right]^{\top}$$
\[
\mathcal{T}_{t} = \begin{bmatrix}
D & \mathbf{0} & \dots & \mathbf{0}  \\
CB & D & \dots & \mathbf{0}  \\
\vdots & \ddots & \ddots & \vdots \\
CA^{t-2}B & CA^{t-3}B & \dots & D
\end{bmatrix}
\]

Using the state-space model, $Y_f$ can be expressed as a linear combination of $X$ and $U_f$:
\begin{align}
\label{eq:Y_f_first_eq}
    Y_f = \mathcal{O}_{T_f}X+\mathcal{T}_{T_f}U_f
\end{align}
and $X$ can be expressed as a linear combination of $U_p$ and $Y_p$:
$$X = A^{T_{p}}\mathcal{O}_{T_p-1}^\dagger(Y_p-\mathcal{T}_{T_p-1}U_p)+\mathcal{C}_{T_p-1}U_p$$
where, since $T_p \geq \ell$, $X$ is unique, and
\[
\mathcal{C}_{T_p-1} = \begin{bmatrix}
A^{T_p-1}B & A^{T_p-2}B & \dots & AB & B \\
\end{bmatrix}
\]
Substituting this expression for $X$ into \eqref{eq:Y_f_first_eq} we conclude that 
$Y_f$ can be expressed as a linear combination of $U_p$, $U_f$, and $Y_p$ and therefore removing $Y_f$ from $\mathcal{H}$ does not alter its rank.
\subsection{Proof of Theorem \ref{thm:no tsvd}}
We use the following notation for this proof:
$$
\delta_h = \tilde{h}-h 
$$
 An initial bound is obtained on $\|\mathbf{\tilde{y}_{pred}} -\mathbf{y_{pred}} \|_2$ as follows:
 \begin{equation}\label{eq:th1_ini_bound}
\begin{aligned}
    & \|\mathbf{\tilde{y}_{pred}} -\mathbf{y_{pred}} \|_2 =\|\tilde{Y}_f\tilde{H}_1^{\dagger}\tilde{h}-Y_fH_1^{\dagger}h\|_2 \\
    & \quad = \|(Y_f+\Delta_2)\tilde{H}_1^{\dagger} (h+\delta_h)  -Y_fH_1^\dagger h \|_2 \\
    & \quad = \|Y_f(\tilde{H}_1^{\dagger}-H_1^\dagger)h+ Y_f\tilde{H}_1^{\dagger} \delta_h +\Delta_2\tilde{H}_1^{\dagger} \tilde{h} \|_2 \\
    & \quad \leq \|Y_f(\tilde{H}_1^{\dagger}-H_1^\dagger)h\|_2+ \|Y_f\tilde{H}_1^{\dagger} \delta_h\|_2 +\|\Delta_2\tilde{H}_1^{\dagger} \tilde{h} \|_2 \\
    & \quad \leq \|Y_f\|_F\|\tilde{H}_1^{\dagger}-H_1^\dagger\|_F\|h\|_2+ \|Y_f\|_F\|\tilde{H}_1^{\dagger}\|_F\|\delta\|_2 \\
    &\quad +\|\Delta_2\|_F\|\tilde{H}_1^{\dagger}\tilde{h}\|_F 
\end{aligned}
\end{equation}
where we have employed the triangle inequality and the fact that the 2-norm is less than the Frobenius norm. We bound $\|Y_f\|_F$ in the above expression as follows:
\begin{align}
    & \|Y_f-\tilde{Y}_f\|_F = \|\Delta_2\|_F \nonumber\\
    & \quad \Rightarrow \|Y_f\|_F-\|\tilde{Y}_f\|_F \leq \|\Delta_2\|_F \nonumber\\
    & \quad \Rightarrow \|Y_f\|_F \leq \|\tilde{Y}_f\|_F + \|\Delta_2\|_F \label{eq:Y_f_bound}
\end{align}
Similarly:
\begin{align}
    & \quad \|h\|_2 \leq\|\tilde{h}\|_2 +\|\delta\|_2\label{eq:h_bound}
\end{align}
Utilizing \cite[Theorem 3.3]{Stewart1977}, we bound $\|\tilde{H}_1^{\dagger}-H_1^\dagger\|_F$ as follows:
\begin{align}
    & \|\tilde{H}_1^{\dagger}-H_1^\dagger\|_F \leq \sqrt{2}\max\{\|H_1^\dagger\|_2^2, \|\tilde{H}_1^{\dagger}\|_2^2 \}\|\Delta_1\|_F
\label{eq:stewart_bound_1}
\end{align}
We proceed to show that  $\max\{\|H_1^\dagger\|_2^2, \|\tilde{H}_1^{\dagger}\|_2^2\}\leq \sigma_{\text{sq}}(\tilde{H}_1)$. To this end, we write:
\begin{align}
    & \max\{\|H_1^\dagger\|_2^2, \|\tilde{H}_1^{\dagger}\|_2^2\} = \max\{\sigma_{\max}(H_1^\dagger)^2, \sigma_{\max}(\tilde{H}_1^{\dagger})^2\} \nonumber \\
    & \quad = \max\left\{\left(\frac{1}{\sigma_{\min}(H_1)}\right)^2,\left(\frac{1}{\sigma_{\min}(\tilde{H}_{1})}\right)^2\right\} \nonumber
\end{align}
To bound $\left(\frac{1}{\sigma_{\min}(H_1)}\right)^2$, we note that $H_1$ is a rank-$r$ matrix based on lemma \ref{lemm:H_1 rank}, therefore $\sigma_{\min}(H_1) = \sigma_{r}(H_1)$. Now, using Weyl's inequality \cite[Proposition 1]{Vu_2021}, we obtain:
\begin{equation}\label{eq:Weyl_1}
\begin{aligned}
    & |\sigma_r(\tilde{H}_1) - \sigma_r(H_1) | \leq \|\Delta_1\|_2 \\
    &\quad \Rightarrow |\sigma_r(\tilde{H}_1)| - |\sigma_r(H_1)| \leq \|\Delta_1\|_F \\
    &\quad \Rightarrow \sigma_r(\tilde{H}_1) - \sqrt{pT_pM}N \leq \sigma_r(H_1) 
\end{aligned}
\end{equation}
where we have used Lemma \ref{lemm:Lemma_1} and the fact that the 2-norm is less than the Frobenius norm. 
Based on our assumption that $\delta_{\text{SN}}(\tilde{H}_1) = \sigma_{r}(\tilde{H}_1) - \sqrt{pT_pM}N$ is positive, we can write:
\begin{equation}
\label{eq:sigma_r_sq_bound}
     \left(\frac{1}{\delta_{\text{SN}}(\tilde{H}_1)}\right)^2 \geq \left(\frac{1}{\sigma_{r}(H_1)}\right)^2 = \left(\frac{1}{\sigma_{\text{min}}(H_1)}\right)^2 
\end{equation}
and consequently:
\begin{align}
    &\quad\max\left\{\left(\frac{1}{\sigma_{\min}(H_1)}\right)^2,\left(\frac{1}{\sigma_{\min}(\tilde{H}_{1})}\right)^2\right\} \nonumber\\
    &\quad \leq \max\left\{\left(\frac{1}{\delta_{\text{SN}}(\tilde{H}_1)}\right)^2,\left(\frac{1}{\sigma_{\min}(\tilde{H}_{1})}\right)^2\right\} = \sigma_{\text{sq}}(\tilde{H}_1)\nonumber
\end{align}
Hence:
\begin{align}
    & \|\tilde{H}_1^{\dagger}-H_1^\dagger\|_F \leq \sqrt{2}\sigma_{\text{sq}}(\tilde{H}_1) \|\Delta_1\|_F \label{eq:Stewart_bound_2}
\end{align}
Combining \eqref{eq:Y_f_bound}, \eqref{eq:h_bound}, and \eqref{eq:Stewart_bound_2} allows us to rewrite \eqref{eq:th1_ini_bound} as:
\begin{align}
    &  \|\mathbf{\tilde{y}_{pred}} -\mathbf{y_{pred}} \|_2 \leq \left(\|\tilde{Y}_f\|_F+\|\Delta_2\|_F\right)\sqrt{2}\sigma_{\text{sq}}(\tilde{H}_1)\|\Delta_1\|_F\nonumber\\
    &\quad \left(\|\tilde{h}\|_2+\|\delta\|_2\right)+ \left(\|\tilde{Y}_f\|_F+\|\Delta_2\|_F\right)\|\tilde{H}_1^{\dagger}\|_F\|\delta\|_2 \nonumber\\   
    &\quad +\|\Delta_2\|_F\|\|\tilde{H}_1^{\dagger}\tilde{h}\|_F \nonumber
\end{align}
Applying Lemma 1, the result follows. 

\subsection{Proof of Theorem \ref{thm:tsvd}}
Similarly to the previous proof:
\begin{align}
\label{eq:ini_bound_hat}
    & \|\mathbf{\hat{y}_{pred}} - \mathbf{y_{pred}}\|_2 = \|\hat{Y}_f\hat{H}_1^\dagger\tilde{h}-Y_fH_1^\dagger h\|_2 \nonumber\\
    &\quad \leq \|Y_f\|_F\|\hat{H}_1^\dagger - H_1^\dagger\|_F\|h\|_2  + \|\hat{Y}_f - Y_f\|_F\|\hat{H}_1^\dagger\|_F\|h\|_2 \nonumber\\
    &\quad  + \|\hat{Y}_f\hat{H}_1^\dagger\|_F\|\delta\|_2
\end{align}
Similar to  \eqref{eq:Stewart_bound_2}, we have that:
$ \|\hat{H}_1^\dagger - H_1^\dagger\|_F \leq \sqrt{2}\sigma_{\text{sq}}(\hat{H}_1)\|\hat{H}_1-H_1\|_F.$
Since $\hat{H}$ is a rank-$r$ matrix by construction, and because ${\delta_{\text{SN}}(\hat{H}_1)>0}$, we have that $\hat{H}_1$ is also rank-$r$. This implies that $\sigma_r(\hat{H}_1) = \sigma_{\min}(\hat{H}_1)$, so we simplify $\sigma_{\text{sq}}(\hat{H}_1)$:
\begin{align}
    & \sigma_{\text{sq}}(\hat{H}_1) =  \max\left\{ \left(\frac{1}{\delta_{\text{SN}}(\hat{H}_1)}\right)^2, \left(\frac{1}{\sigma_{\text{min}}(\hat{H}_1)}\right)^2\right\} \nonumber\\
    & \quad = \max\left\{\left(\frac{1}{\delta_{\text{SN}}(\hat{H}_1)}\right)^2, \left(\frac{1}{\sigma_{{r}}(\hat{H}_1)}\right)^2\right\} =\left(\frac{1}{\delta_{\text{SN}}(\hat{H}_1)}\right)^2 \nonumber
\end{align}
Next, we calculate a bound on $\|\hat{H}_1 -H_1\|_F$:
\begin{align}
\label{eq:hat-clean}
    & \|\hat{H}_1 -H_1\|_F = \|\hat{H}_1 -(\tilde{H}_1-[\mathbf{0}_{mT\times M}^{\intercal} \Delta_1^{\intercal}]^{\intercal})\|_F\nonumber\\
    & \quad \leq \|\hat{H}_1-\tilde{H}_1\|_F+\|\Delta_1\|_F
\end{align}
Therefore:
\begin{align}
    \label{eq:enhanced_bound}
    & \|\hat{H}_1^\dagger - H_1^\dagger\|_F \leq \sqrt{2}\left(\frac{1}{\delta_{\text{SN}}(\hat{H}_1)}\right)^2\left(\|\hat{H}_1-\tilde{H}_1\|_F+\|\Delta_1\|_F\right)
\end{align}
We bound $\|\hat{Y}_f - Y_f\|_F$ similar to $\|\hat{H}_1 -H_1\|_F$:
\begin{align}
\label{eq:delta_Y_f_hat_bound}
    & \|\hat{Y}_f - Y_f\|_F \leq  \|\hat{Y}_f-\tilde{Y}_f\|_F+\|\Delta_2\|_F 
\end{align}
Using \eqref{eq:Y_f_bound}, \eqref{eq:h_bound}, \eqref{eq:enhanced_bound}, and \eqref{eq:delta_Y_f_hat_bound} we rewrite \eqref{eq:ini_bound_hat} as:
\begin{align}
    &\|\mathbf{\hat{y}_{pred}} - \mathbf{y_{pred}}\|_2 \leq \left(\|\tilde{Y}_f\|_F + \|\Delta_2\|_F\right)\sqrt{2}\left(\frac{1}{\delta_{\text{SN}}(\hat{H}_1)}\right)^2\nonumber\\
    & \quad \left(\|\hat{H}_1-\tilde{H}_1\|_F+\|\Delta_1\|_F\right)\left(\|\tilde{h}\|_2+\|\delta\|_2\right)\nonumber\\
    & \quad + \left(\|\hat{Y}_f-\tilde{Y}_f\|_F+\|\Delta_2\|_F\right)\|\hat{H}_1^\dagger\|_F\left(\|\tilde{h}\|_2+\|\delta\|_2\right)\nonumber\\
    & \quad +\|\hat{Y}_f\hat{H}_1^\dagger\|_F\|\delta\|_2\nonumber
\end{align}
Applying Lemma \ref{lemm:Lemma_1} proves the result. 
\begin{remark}
    In place of  \eqref{eq:bound_2}, a different upper bound can be obtained if we consider the approach in \cite[Section 6.3]{Vu_2021} to bound $\|\hat{H}_1-H_1\|_F$:
    \begin{align}
        & \|\hat{H}_1-H_1\|_F \leq \|\Delta_1\|_F+\|P_{U_2}\Delta_1P_{V_2}\|_F + 2(1+\sqrt{2})\|\Delta_1\|_F \nonumber\\
        & \quad \min\left\{\frac{2}{\sigma_r(H_1)}\|\Delta_1\|_F,1\right\} \nonumber
    \end{align}
    where $P_{U_2}$ and $P_{V_2}$ are orthogonal matrices. Multiplication by the orthogonal matrices preserves the Frobenius norm. Based on this and \eqref{eq:sigma_r_sq_bound} we have:
    \begin{align}
    \label{eq:vu_bound}
    \|\hat{H}_1-H_1\|_F\leq 2\|\Delta_1\|_F + 2(1+\sqrt{2})\|\Delta_1\|_F \nonumber\\
    \min\left\{\frac{2}{\delta_{\text{SN}}(\tilde{H}_1)}\|\Delta_1\|_F,1\right\}
    \end{align}
    However, in Theorem 2, we utilize \eqref{eq:hat-clean} to bound $\|\hat{H}_1-H_1\|_F$ instead of \eqref{eq:vu_bound}, because our numerical studies showed that it provides a tighter bound when the noise is modeled using \eqref{eq:noise_model_1}, \eqref{eq:noise_model_2}.
\end{remark}
\bibliographystyle{unsrt}
\bibliography{main}

\end{document}